\documentclass[11pt,A4paper,fullpage]{article}

\usepackage{todonotes} 

\usepackage{amsfonts}
\usepackage{fullpage}
\usepackage{graphicx,color}
\usepackage{xspace}
\usepackage{enumerate}
\usepackage{amssymb,amsmath}
\usepackage{amsthm}
\usepackage{subfig}

\usepackage[utf8]{inputenc}
\usepackage{times}

\usepackage[utf8]{inputenc}
\usepackage{authblk}
\usepackage{hyperref}

\usepackage{mathtools}

\newtheorem{theorem}{Theorem}

\newtheorem{lemma}[theorem]{Lemma}

\newtheorem{definition}[theorem]{Definition}

\def\multiset#1#2{\ensuremath{\left(\kern-.3em\left(\genfrac{}{}{0pt}{}{#1}{#2}\right)\kern-.3em\right)}}

\newcommand{\mch}[2]{
\left.\mathchoice
  {\left(\kern-0.48em\binom{#1}{#2}\kern-0.48em\right)}
  {\big(\kern-0.30em\binom{\smash{#1}}{\smash{#2}}\kern-0.30em\big)}
  {\left(\kern-0.30em\binom{\smash{#1}}{\smash{#2}}\kern-0.30em\right)}
  {\left(\kern-0.30em\binom{\smash{#1}}{\smash{#2}}\kern-0.30em\right)}
\right.}

\newcommand{\wa}{{\sc Weighting}\xspace}
\newcommand{\Opt}{{\sc Opt}\xspace}

\newcommand{\acc}{\ensuremath{\textrm{Cr}}\xspace}

\newcommand{\comp}{\ensuremath{\textrm{Cr}}\xspace}

\newcommand{\mincycl}{\textsc{MinCyclic}$(n,k,H)$\xspace}
\newcommand{\search}{\textsc{Search}$(n,k,H)$\xspace}

\newcommand{\cSearch}{\textsc{ContiniousSearch}($k,H$) \xspace}

\newcommand{\Identify}{\textsc{Identify}($m,H$) \xspace}
\newcommand{\Range}{\textsc{Find}($k,H$)\xspace}

\newcommand{\CWeighting}{\textsc{C-Weighting}}

\begin{document}
\title{R\'enyi-Ulam Games and Online Computation with Imperfect Advice}

\author[1]{Spyros Angelopoulos}
\author[2]{Shahin Kamali}

\date{}

\affil[1]{CNRS and LIP6-Sorbonne University, Paris, France. {\tt spyros.angelopoulos@lip6.fr}} 
\medskip

\affil[2]{York University, Canada. {\tt kamalis@yorku.ca} }

\maketitle

\begin{abstract}
We study the nascent setting of online computation with {\em imperfect} advice, in which the online algorithm is enhanced by some prediction encoded in the form of a possibly erroneous binary string. The algorithm is oblivious to the advice error, but defines a desired {\em tolerance}, namely an upper bound on the number of erroneous advice bits it can tolerate. This is a model that generalizes the {\em untrusted} advice model [Angelopoulos et al. ITCS 2020], in which the performance of the algorithm is only evaluated at the extreme values of error (namely, if the advice has either no errors, or if it is generated adversarially).

In this work, we establish connections between games with a lying responder, also known as {\em R\'enyi-Ulam games}, and the design and analysis of online algorithms with imperfect advice. Specifically, we demonstrate how to obtain upper and lower bounds on the competitive ratio for well-studied online problems such as time-series search, online bidding, and fractional knapsack. Our techniques provide the first lower bounds for online problems in this model. We also highlight and exploit connections between competitive analysis with imperfect advice and fault-tolerance in multiprocessor systems. Last, we show how to waive the dependence on the tolerance parameter, by means of resource augmentation and robustification.

\medskip

\noindent
{\bf Keywords} \ Online computation, noisy queries, R\'enyi-Ulam games, beyond worst-case analysis, fault-tolerant algorithms.

\end{abstract}

\section{Introduction}
\label{sec:introduction}

Online computation, and {\em competitive analysis}, in particular, have served as the definitive framework for the theoretical analysis of algorithms in a state of uncertainty. While the early, standard definition of online computation~\cite{SleTar85} assumes that the algorithm has no knowledge in regards to the request sequence, in practical situations the algorithm may indeed have certain limited, but possibly inaccurate such information (e.g., some lookahead, or historical information on typical sequences). Hence the need for more nuanced models that capture the power and limitations of online algorithms enhanced with external information. 

One such approach, within Theoretical Computer Science, is the framework of {\em advice complexity}; see~\cite{DobKraPar09,BocKomKra09,EmekFraKorRos2011}, the survey~\cite{Boyar:survey:2016} and the book~\cite{komm2016introduction}. In the advice-complexity model (and in particular, the {\em tape} model~\cite{DBLP:journals/jcss/BockenhauerKKK17}), the online algorithm receives a string that encodes information concerning the request sequence, and which can help improve its performance. The objective is to quantify the trade-offs between the size of advice (in terms of number of bits), and the competitive ratio of the algorithm. This model places stringent requirements: the advice is assumed to be error free, and may be provided by an omnipotent oracle. Thus, as noted in~\cite{DBLP:books/cu/20/MitzenmacherV20}, this model is mostly of theoretical significance.

A different, and more practical approach, studies the effect of {\em predictions} towards improving the competitive ratio. In this model, the online algorithm is enhanced with some imperfect information concerning the request sequence, without restrictions on its size. One is interested in algorithms whose performance degrades gently as function of the prediction {\em error}, and specifically perform well if the prediction is error free (what is called the {\em consistency} of the algorithm), but also remain robust under any possible error (what is called the {\em robustness} of the algorithm). This line of research was initiated with the works~\cite{DBLP:conf/icml/LykourisV18} and~\cite{NIPS2018_8174}, and a very large number of online problems have been studied under this model (see, e.g., the survey~\cite{DBLP:books/cu/20/MitzenmacherV20} and the online collection~\cite{predictionslist}).

A combination of the advice complexity and prediction models is the {\em untrusted advice} model, introduced in~\cite{DBLP:conf/innovations/0001DJKR20}. Here, some of the advice bits may be erroneous, and the algorithm's performance is evaluated at two extreme situations, in regards to the advice error. At the one extreme, the advice is error-free, whereas at the other extreme, the 
advice is generated by a (malicious) adversary who aims to maximize the performance degradation of the algorithm. Using the terminology of algorithms with predictions, these two competitive ratios are called consistency and robustness, respectively. The objective is to identify algorithms that are {\em Pareto-efficient}, and ideally {\em Pareto-optimal}, for these two extreme measures.  Several online problems have been studied recently within this framework of Pareto-optimality (both within the untrusted advice and the predictions models); see, e.g.,~\cite{wei2020optimal,li2021robustness,lee2021online,DBLP:conf/innovations/000121, DBLP:conf/aaai/0001K21}.

\subsection{Online computation with imperfect advice}

In this work, we focus on a nascent model in which the advice can be imperfect. The starting observation is that the Pareto-based framework of untrusted advice only focuses on extreme competitive ratios, namely the consistency and the robustness. A more general issue is to evaluate the performance of an online algorithm as function of the advice error. Given an advice string of size $k$, we denote by $\eta\leq k$ the number of erroneous bits. The objective is then to study the power and limitations of online algorithms within this setting, i.e., from the point of view of both upper and lower bounds on the competitive ratio.

Naturally, the algorithm does not know the exact advice error ahead of time. Instead, the algorithm defines an appliction-specific parameter $H \leq k$ which determines the desired {\em tolerance} to errors, or, equivalently, an anticipated upper bound on the advice error. This parameter appears very often in the analysis of games with a lying responder such as {\em R\'enyi-Ulam games}~\cite{RivestMKWS80}, which are of interest to our work, as we will discuss shortly. It is also further motivated by recent works in learning-enhanced online algorithms with {\em weak predictions}, in which the prediction is an upper bound of some pertinent parameter of the input (see e.g., online knapsack with frequency predictions~\cite{knapsack:frequency}, where the prediction is an upper bound on the size of items that appear online). Our objective is to quantify the tradeoffs between advice size, tolerance and competitive ratio, both from the point of upper and lower bounds.

A different interpretation of the imperfect advice model treats each advice bit as a response to a {\em binary query} concerning the input. Hence, one may think of a $k$-bit advice string as a prediction elicited in the form of $k$ binary queries, not all of which may receive correct responses. Queries are known to help improve the performance of approximation algorithms in ML applications; see, e.g, clustering with noisy queries~\cite{NIPS2017_7161}, in which a query asks whether two points should belong in the same cluster, and where each query receives a correct response with probability $p$ that is known to the algorithm. In this work, we study the power, but also the limitations of online algorithms with noisy queries. However, unlike~\cite{NIPS2017_7161}, we do not rely on any probabilistic assumptions concerning the query responses. To our knowledge, the imperfect advice model (in particular, its binary query-based interpretation) has only been applied to the problems of {\em contract scheduling}~\cite{DBLP:conf/aaai/0001K21}, and time-series search~\cite{DBLP:conf/aaai/0001K22}, and solely from the point of view of upper bounds.

\subsection{Contribution}
\label{subsec:contribution}

We establish connections between games with a lying responder and the design and analysis of online algorithms with imperfect advice. Namely, we show how to leverage results from the analysis of {\em R\'enyi-Ulam} games, and obtain both positive and negative results on the competitive analysis. We apply these tools towards well-studied problems such as time-series search, online bidding, and online fractional knapsack. Our results improve the known upper bounds for these problems, where such results were already known, but also provide the first lower bounds on the competitive ratio of online problems in this setting.

More precisely, we begin as a warm-up with the time-series search problem in Section~\ref{sec:ts}, which illustrates how these techniques can help us improve upon the results of~\cite{DBLP:conf/aaai/0001K22}; we also show how to evaluate the competitive ratios, using approximations based on the binary entropy function. In Section~\ref{sec.bidding}, we study a more complex application, namely the online bidding problem, first studied in~\cite{DBLP:conf/innovations/0001DJKR20} in the context of untrusted advice. Here, the crucial part is establishing near-optimal lower bounds. We achieve this by formulating a multi-processor version of online bidding in $l \leq 2^k$ processors, in which a certain number of processors may be faulty; we then relate the competitive ratio of this problem to the imperfect advice setting, by relating fault-tolerance in the processor level, to the inherent error in R\'enyi-Ulam games. In Section~\ref{sec:fractional.knapsack} we study the online fractional knapsack problem. Here, we present an algorithm whose competitive ratio converges to 1 at a rate exponential to $k$, as long as $H< k/2$	. We also present a near-matching lower bound that shows that our algorithm is close-to-optimal. For the upper bound, the crux is to use to allocate queries so as to approximate 
two appropriately defined parameters of the instance. For the lower bound, we use an information theoretic argument. Specifically, we show a reduction from R\'enyi-Ulam games: if there existed an algorithm of competitive ratio better than a certain value, one could play the game beyond the theoretical performance bound, which is a contradiction.

As explained above, the parameter $H$ expresses the algorithm's desired tolerance to errors, and is thus application-specific. In Section~\ref{sec:extensions} we argue that the results are useful even in settings in which the precise tolerance is not known ahead of time. We accomplish this in two different ways: First, by {\em resource-augmentation} arguments, i.e., by comparing the performance of an algorithm with perfect (error-free) advice of size $k$ to that of an algorithm with $l>k$ advice bits but potentially very high advice error. Second, by {\em robustifying} the algorithm, namely by requiring that the algorithm performs well even if the error happens to exceed the tolerance parameter.

The techniques we develop can be applicable to other online problems. Specifically, our approach to the online bidding problem defines the following general framework:
For upper bounds, one would aim to define a collection of ``candidate'' algorithms that are closely ranked in terms of their worst-case performance. Then the advice can be used so as to select a suitable candidate from this collection that is close to the best-possible. For lower bounds, one would aim to show that in any collection of candidate algorithms, the erroneous queries may have to always return a solution sufficiently far, in terms of ``rank'', from the best one; then one needs to relate the concept of ``rank'' to performance, from a lower-bound point of view. This last part highlights connections between an online problem with imperfect advice, and its fault-tolerant version in a parallel system (with no advice). On the other hand, our approach to the time-series and fractional knapsack problems illustrate another general technique: For upper bounds, one should identify some important parameters of the problem, then allocate the queries appropriately so as to approximate them in the presence of response errors. For lower bounds, information-theoretic arguments should establish a reduction from a R\'enyi-Ulam game to the online problem.

\section{Games with a lying responder}
\label{sec:core}

We review some core results related to games with a lying responder which will be in the heart of the analysis of online problems with imperfect advice. We are interested, in particular, in~\cite{RivestMKWS80}, which studied games between a {\em questioner} and a {\em responder}, related to an unknown value $x$ drawn from a domain ${\cal D}$. The questioner may ask general queries of the form ``is $x$ in $S$'', where $S$ is some subset of ${\cal D}$, and which are called {\em subset} queries. The upper bounds of~\cite{RivestMKWS80} hold even if the questioner asks much simpler queries, namely {\em comparison} queries of the form ``is $x$ at most $a$'', for some given $a$. Both the upper and lower bounds in~\cite{RivestMKWS80} are expressed in terms of partial sums of binomial coefficients. Formally, we define:
\[
  \mch{N}{m}   := \sum_{j=0}^m {N \choose j}, \  \textrm{for $m \leq N$}.
\]

We are interested, in particular, in the following game played over a continuous space:

\bigskip

\noindent \textbf{\cSearch game:} In this game, $x$ is a real number with $x\in {\cal D}=(0,1]$, and the questioner asks $k$ queries, at most $H$ of which may receive erroneous responses. The objective of the questioner is to find an interval $I_x$ such that $x\in I_x$ and $|I_x|$ is minimized.

\begin{lemma}~\cite{RivestMKWS80}\label{lem:csearch}
Any questioner's strategy for \cSearch with $H \leq k/2$ is such that $|I_{x}| \geq \multiset{k}{H}/2^k$. Moreover, for $H\leq k/2$, there is a strategy, named \CWeighting, that uses comparison queries and outputs an interval $I_{W,x}$ with $|I_{W,x}| \leq \multiset{k-H}{H}/2^{k-H}$.
\end{lemma}

The following game will be useful in our analysis of online time-series and fractional knapsack. 
\bigskip

\noindent \textbf{\Range game:}
In this game, given $k$ and $H\leq k/2$, and ${\cal D}=\{1,\ldots,m\}$, the objective of to find an unknown $x \in {\cal D}$, using $k$ queries, up to $H$ of which may be answered incorrectly. 

The proof of the following theorem is direct from Lemma~\ref{lem:csearch}:
\begin{theorem}\label{th:range}
The largest positive integer $\mu(k,H)$ such that a questioner can identify any number $x \in \{1,2,\ldots, \mu(k,H)\}$ in the \Range game is such that 
$2^{k-H}/\multiset{k-H}{H} \leq \mu(k,H) \leq   2^k/\multiset{k}{H}$. 
\end{theorem}

We define two further games that will be of interest to our analysis. The first is related to searching in cyclic permutations, and will be useful in the upper-bound analysis of online bidding.
\bigskip

\noindent
{\bf \mincycl game:} \ Given an array $A[0\ldots n-1]$ whose elements are an unknown cyclic permutation of $\{0, \ldots ,n-1\}$, the objective is to use $k$ queries, at most $H\leq k/2$ of which can be erroneous, so as to output an index of the array whose element is as small as possible.

\begin{theorem}[Appendix]
There is a questioner's strategy for \mincycl based on $k$ comparison queries that 
outputs an index $j$ such that $A[j] \leq \lceil n\mch{k-H}{H}/2^{k-H} \rceil$, for all $H \leq k/2$.
\label{thm:cyclic.upper}
\end{theorem}

Last, we define a game that is related to searching in general permutations, and it will be useful in establishing lower bounds on the competitiveness of online bidding.
\bigskip

\noindent
{\bf \search game:} \ Given an array, $A[0, \ldots ,n-1]$ whose elements are an unknown permutation
of $\{0, \ldots ,n-1\}$, the objective is to use $k$ queries, at most $H$ of which can be erroneous, so as to output an index of the array whose element is as small as possible. 
 
\bigskip

\begin{theorem}[Appendix]
For any questioner's strategy for the \search game, there is a responder's strategy such that if $e$ is the element of $A$ that is returned, then $A[e] \geq \lfloor n\mch{k}{H}/2^{k} \rfloor$. 
\label{thm:search.lower}
\end{theorem}

\section{A warm-up: Online time-series search}
\label{sec:ts}

The {\em online (time series) search} problem formulates a simple, yet fundamental setting in decision-making under uncertainty. 
In this problem, a player must sell an indivisible asset within a certain time horizon, e.g., within a certain number of days $d$, that is unknown to the player. On each day $i$, a {\em price} $p_i$ is revealed, and the player has two choices: either accept the price, and gain a profit $p_i$ (at which point the game ends), or reject the price (at which point the game continues to day $i+1$). If the player has not accepted a price by day $d$, then it accepts by default the last price $p_d$. The competitive ratio of the player's algorithm is the worst-case ratio, over all price sequences, of the maximum price in the sequence divided by the price accepted by the player.

The problem was introduced and studied in~\cite{el2001optimal} that gave a simple, deterministic algorithm that achieves a competitive ratio equal to $\sqrt{M/m}$, where $M,m$ are upper and lower bounds on the maximum and minimum price in the sequence, respectively, and which are assumed to be known to the algorithm. This bound is optimal for deterministic algorithms. 
Time-series search is a basic paradigm in online financial optimization, and several variants and generalizations have been studied~\cite{damaschke2009online, lorenz2009optimal, xu2011optimal, clemente2016advice}; see also the survey~\cite{li2014online}. The problem has also been used as a case study for evaluating several performance measures of online algorithms, including measures alternative to competitive analysis~\cite{boyar2014comparison,ahmad2021analysis}. 

Time-series search was recently studied under the imperfect advice framework in~\cite{DBLP:conf/aaai/0001K22}, who showed an upper bound of  $(M/m)^{2^{2H-k/2}}$ on the competitive ratio with $k$-bit advice and tolerance $H$, under the assumption that $H\leq k/4$. Note that no upper bound is known for
$H \in (k/4, k/2]$. If the advice is error-free, i.e., in the advice-complexity model, then a tight bound on the competitive ratio equal to $(M/m)^\frac{1}{2^k+1}$ is due to~\cite{clemente2016advice}.

We show the following result, as an application of the \Range game discussed in Section~\ref{sec:core}.

\begin{theorem}
Consider the online time series search problem, with imperfect advice of size $k$ and tolerance $H\leq k/2$. There is an algorithm that uses $k$ comparison queries, and that has competitive ratio at most 
$(M/m)^\frac{1}{U+1}$, where $U={\lfloor 2^{k-H}/\multiset{k-H}{H}}\rfloor$, for any $H\leq k/2$. In contrast, every (deterministic) algorithm based on $k$ subset queries has competitive ratio less than 
$(M/m)^\frac{1}{L+1}$, where $L={\lceil 2^k/\multiset{k}{H}}\rceil$.
\label{thm:ts}
\end{theorem}

\begin{proof}
We first show the upper bound. Let  $a_1, \ldots a_U, r$ be defined such that 
$
r=\frac{a_1}{m} = \frac{a_2}{a_1} = \ldots = \frac{a_{U}}{a_{U-1}}=\frac{M}{a_U},
$
hence $r=(M/m)^{1/(U+1)}$. The algorithm uses $k$ comparison queries so as to find the best {\em reservation} price, in the set $\{a_i\}_{i=1}^U$, i.e., the threshold $p$ above which the algorithm will always accept a price in the sequence. This follows from Theorem~\ref{th:range}, since $U\leq 2^{k-H}/\multiset{k-H}{H}$. The algorithm then uses $p$ as its reservation price, namely it accepts the first price in the request sequence that is at least as large as $p$. From the definition of the set $\{a_i\}_{i=1}^U$, it easily follows that this algorithm has competitive ratio at most $r$, which completes the proof of the upper bound.

We now show the lower bound. By way of contradiction, suppose that there is an algorithm $A$ for time-series search with $k$-bit imperfect advice, and of competitive ratio less than $C=(M/m)^\frac{1}{L+1}$.  We will show that $A$ could then be used in the \Range game so as to identify, using $k$ queries, an unknown value in $\{1,\ldots,L+1\}$, which is a contradiction to the upper bound of Theorem~\ref{th:range}.

To arrive at the contradiction, define $a_1, \ldots ,a_{L}$ such that $r'=\frac{a_1}{m}=\frac{a_2}{a_1}=\ldots =\frac{a_{L}}{a_{L-1}}=\frac{M}{a_{L}}$, hence  $r'=(M/m)^\frac{1}{L+1}=C$.
Consider a game between the online algorithm $A$ and the adversary, in which the request sequences consist of prices in $\{m,a_1, \ldots, a_{L},M\}$. More precisely, consider the set of request sequences of the form $\sigma_i=m, a_1, \ldots ,a_i,m$, for all $i \in [1,L+1]$, 
where $a_{L+1}$ is defined to be equal to $M$. In $\sigma_i$, $A$ must accept price $a_i$ (the last request in the sequence) to be strictly less than $C$-competitive. Equivalently, $A$ uses $k$ queries
with at most $H$ errors, and finds $a_i$ in the set $\{a_j\}_{j=1}^{L+1}$, which contradicts Theorem~\ref{th:range}.
\end{proof}

\subsection{Comparison of the bounds}
\label{subsec:ts.comparisons}

In order to compare the upper and lower bounds of Theorem~\ref{thm:ts}, we need to be able to evaluate the partial sum of binomial coefficients. Since this partial sum does not have a closed form, we will rely on the following useful approximation from~\cite{macwilliams1977theory}. Let ${\cal H}$ denote the {\em binary entropy} function. Then

\begin{equation}
\frac{2^{N{\cal H}(\frac{m}{N})}}{\sqrt{8m(1-\frac{m}{N})}}
\leq \mch{N}{m} \leq 2^{N{\cal H}(\frac{m}{N})}, \quad \textrm{for $0<m<N/2$.} 
\label{eq:amazing}
\end{equation}

We will also use the following property of the binary entropy function
\begin{equation}
4p(1-p) \leq {\cal H}(p) \leq (4p(1-p))^{1/\ln 4}, \ \textrm{ for all $p \in (0,1)$}.
\label{eq:entropy}
\end{equation}

We first show that the algorithm of Theorem~\ref{thm:ts} improves upon the one of~\cite{DBLP:conf/aaai/0001K22}. First, note that~\cite{DBLP:conf/aaai/0001K22} assumes that $H\leq k/4$, whereas Theorem~\ref{thm:ts} applies to all $H\leq k/2$. Furthermore, we improve on the competitive ratio for all values of $H$ and $k$. For this, it suffices to show that 
${\multiset{k-H}{H}}/{2^{k-H}}<2^{2H-k/2}$, which, from~\eqref{eq:amazing} holds if 
$2^{(k-H)({\cal H}(\frac{H}{k-H})-1)}<2^{2H-k/2}$, or equivalently $(k-H)({\cal H}(\frac{H}{k-H})-1)<2H-k/2$. Let $\tau$ be such that $\tau=H/k$ (hence $\tau\leq 1/2$), then the latter is equivalent to showing that 
${\cal H}(\frac{\tau}{1-\tau}) <\frac{1+2\tau}{2-2\tau}$.
Using~\eqref{eq:entropy}, it suffices to show that
\[
(\frac{4\tau(1-2\tau)}{(1-\tau)^2})^{1/\ln 4} <\frac{1+2\tau}{2-2\tau},
\]
which holds for all $\tau\leq 1/2$.

Next, we investigate how close the upper and lower bounds of Theorem~\ref{thm:ts} are to each other. Recall that the bounds are of the form 
$(M/m)^{1/(U+1)}$, and $(M/m)^{1/(L+1)}$. Using~\eqref{eq:amazing}, and ignoring for simplicity the floors and ceilings, we obtain that 
\[
U\geq 2^{k(1-\tau)(1-{\cal H}(\frac{\tau}{1-\tau}))} \ \textrm{ and } L\leq \sqrt{8k\tau(1-\tau)}2^{k(1-{\cal H}(\tau))}.
\]
The above inequalities, along with~\eqref{eq:entropy} show that the upper and lower bounds are very close to each other, since for any fixed value of $\tau$, we have that 
$U\geq 2^{\Theta(k)}$ and $L \leq 2^{\Theta(k)}$.

\section{Online bidding}
\label{sec.bidding}

Online bidding was introduced in~\cite{ChrKen06} as a canonical problem for formalizing doubling-based strategies in online and offline optimization problems, such as searching for a target on the line, minimum latency, and hierarchical clustering. In this problem, a player 
wants to guess a hidden, unknown real value $u \geq 1$. To this end, the player defines an (infinite) sequence $X=(x_i)$ of positive, increasing {\em bids}, which is called its {\em strategy}. The {\em cost of discovering} the hidden value $u$ using the strategy $X$, denoted by $c(X,u)$, is defined to be equal to $\sum_{i=1}^{j_u} {x_i}$, where $j_u$ is such that $x_{j_u-1}<u\leq x_{j_u}$. Hence one  naturally defines the competitive ratio of the bidder's strategy $X$ as
$
\comp(X)= \sup_{u} \frac{c(X,u)}{u}.
$

In the standard version of the problem, i.e, assuming no advice, the doubling strategy $x_i=2^i$ achieves optimal competitive ratio equal to 4.
Online bidding was also studied under the untrusted advice model in~\cite{DBLP:conf/innovations/0001DJKR20}, which gave bounds on the consistency/robustness tradeoffs. The problem is also related to contract scheduling, studied in~\cite{DBLP:conf/aaai/0001K21}, see also the discussion in 
Section~\ref{subsec:bidding.comparisons}.

\subsection{Online bidding with imperfect advice}

\subsubsection{Upper bound}
\label{subsec:bidding.upper}

The idea behind the upper bound is as follows. We will consider bidding sequences from a space of 
$2^k$ geometrically-increasing sequences (see Definition~\ref{def:bunch.old}). In the ideal situation of perfect advice, the $k$ advice bits could be used to identify the best strategy in this space. In the presence of advice errors, we will show how to exploit the cyclic structure of this space, in conjunction with our upper bound for the {\sc MinCyclic} game (Theorem~\ref{thm:cyclic.upper}), so as to find a strategy that is not too far from the optimal. 

We first define the space of geometrically-increasing bidding sequences. 

\begin{definition}
For given $b>1$, and $l \in \mathbb{N}^+$ define ${\cal X}_{b,l}$ as the set of bidding sequences
$\{X_0, \ldots X_{l-1}\}$, in which $X_i=(b^{i+jl})_{j=0}^\infty$, for all $i\in [0,l-1]$.
\label{def:bunch.old}
\end{definition}

From the definition of ${\cal X}_{b,l}$, it is easy to see that for any potential target $u$, there is a cyclic permutation $\pi$ of $\{0, \ldots l-1\}$ which determines an ordering of the strategies in ${\cal X}_{b,l}$ in terms of their performance.  More precisely, suppose that $X_{\pi(0)}$ is the best sequence that discovers $u$ at least cost, say $C$. Then $X_{\pi(i)}$ discovers $u$ at cost at most
$b^iC$. This property can help us show the following upper bound:


\begin{theorem}[Appendix]
There is a bidding strategy based on $k$ comparison queries of competitive ratio at most
$
\frac{1+U}{2^k}\left(1+\frac{2^k}{1+U}\right) ^{1+\frac{1+U}{2^k}},  
\textrm{where $U=\lceil 2^H \mch{k-H}{H} \rceil$}.
$
\label{thm:noisy.upper}
\end{theorem}

\subsubsection{Lower bound}
\label{subsec:bidding.lower}

The idea behind the lower bound is as follows. With $k$ advice bits, the best one can do is 
choose the best strategy from a set ${\cal X}$ that consists of at most $2^k$ strategies. Note that if the advice were error-free, $|{\cal X}|$ could be as large as $2^k$; however, in the presence of errors, the algorithm may choose to narrow $|{\cal X}|$.

Our approach combines two ideas. The first idea uses the abstraction of the \search game, and the lower bound of Theorem~\ref{thm:search.lower}. This result will allow us to place a lower bound on the rank of the chosen strategy, where the best strategy has rank 0. 
The second idea is to define a {\em measure} that relates how much worse a strategy of rank $j$ in ${\cal X}$ has to be relative to the best strategy in ${\cal X}$. We will accomplish this by appealing to the concepts of {\em parallelism} and {\em fault tolerance}.

 More precisely, given integers $p$, and $\phi$, with $\phi<p$, we define the {\em fault-tolerant parallel bidding} problem, denoted by $\mathtt{FPB}(p,\phi)$, as follows. The player is allowed to run, in parallel, $p$ bidding strategies; however, $\phi$ of these strategies can be {\em faulty},
 in that they never discover the target; e.g., we can think of a fault strategy as one in which the player abruptly stops submitting bids, at some point in time, akin to a ``byzantine'' failure. The cost of discovering a target $u$ is then defined as the minimum cost at which one of the $p-\phi$ non-faulty strategies discovers the target, noting that the faults are dictated by an adversary that aims to maximize this cost. The competitive ratio is defined accordingly.

The next theorem is the main technical result for $\mathtt{FPB}(p,\phi)$, which gives a lower bound on the competitive ratio of any strategy for this problem, as a function of the parameters $p$, $\phi$ and $\alpha_{\bar{X}}$. Here, $\bar{X}$ is defined as the sorted sequence of all bids in the $p$-parallel strategy $X$, in non-decreasing order. Moreover, given a sequence $X$ of positive reals, 
we define $\alpha_X$ to be equal to $\limsup_{i \to \infty} x_i^{1/i}$.

\begin{theorem}[Appendix]
Every $p$-parallel strategy $X$ for $\mathtt{FPB}(p,\phi)$ has competitive ratio 
$
\comp(X) \geq \frac{\alpha_{\bar{X}}^{p+1+\phi}}{\alpha_{\bar{X}}^p-1}.
$
\label{thm:multi-search.alpha.faulty}
\end{theorem}

\begin{proof}[Proof sketch]
We use properties of $p$-parallel strategies so as to show that any such strategy satisfies $\comp(X) \geq \sup_{q} \frac{\sum_{i=0}^{q+\phi+1} \bar{x}_i}{\sum_{i=q}^{q-(p-1)} \bar{x}_{i}}$. We then use Gal's functional theorem~\cite{gal:general} to obtain the result. We refer to Appendix for many technical details.
\end{proof}

We now show how to obtain a lower bound for the problem by combining the above ideas. 
We emphasize a subtle point: unlike error-free advice of size $k$, where one should always choose the best strategy out of a collection of exactly $2^k$ strategies, it is conceivable that, in the presence of errors, this collection could very well be of size $l<2^k$. This is because, as $l$ decreases, so does the effect of errors on the competitive ratio. In other words, we need to establish the result for {\em all} values $l\leq 2^k$, and not only for
$l=2^k$.

\begin{theorem}
For every bidding sequence $X$ and $k$ subset queries in the imperfect advice model, 
we have $\comp(X) \geq \frac{1}{L}(1+L)^{1+1/L}$, where $L=2^k/\mch{k}{H}$.
\label{thm:noisy.lower}
\end{theorem}

\begin{proof}
Every bidding strategy will use the query responses so as to select a strategy from a set 
${\cal X}=\{X_0, \ldots , X_{l-1}\}$ of candidate sequences,  for some $l \leq 2^k$. For a given target value $u$, there is an ordering of the $l$ sequences in 
${\cal X}$ such that $X_{\pi(i)}$ has no worse competitive ratio than $X_{\pi(i+1)}$, namely the permutation orders the sequences in decreasing order of performance. From Theorem~\ref{thm:search.lower}, it follows that the strategy will choose a sequence
$X_{j}$ such that $\pi(j)\geq \lfloor l\mch{k}{H}/2^{k} \rfloor$. The competitive ratio of the selected sequence is at least the competitive ratio of the $l$-parallel strategy defined by 
${\cal X}$, in which up to $\phi_l = \lfloor l\mch{k}{H}/2^{k} \rfloor$ sequences may be faulty. From Theorem~\ref{thm:multi-search.alpha.faulty},
\begin{equation}
\comp(X) \geq \frac{\alpha_{\bar{X}}^{l+1+\phi_l}}{\alpha_{\bar{X}}^l-1}, \quad \textrm{with $\phi_l = \lfloor l\mch{k}{H}/2^{k} \rfloor$}.
\label{eq:nasty}
\end{equation}
We now consider two cases. Suppose first that $l<L$. In this case,
case $\phi_l=0$, and therefore~\eqref{eq:nasty} implies that $\comp(X) \geq \alpha_{\bar{X}}^{l+1}/(\alpha_{\bar{X}}^l-1)$, which is minimized for $\alpha_{\bar{X}}=(l+1)^{1/l}>1$, therefore $\acc(X) \geq \frac{1}{l}(l+1)^{1+1/l}$. This function is decreasing in $l$, and since $l<L$ we have
$
\comp(X) \geq \frac{1}{L}(1+L)^{1+1/L}.
$
Next, suppose that $l \in [L,2^k]$. In this case,~\eqref{eq:nasty} gives 
$
\comp(X) \geq \frac{\alpha_{\bar{X}}^{l(1+1/L)}}{\alpha_{\bar{X}}^l-1}.
$ 
The above expression is minimized for $\alpha_{\bar{X}} =(1+L)^{1/l}$, and by substitution we obtain again
$
\comp(X) \geq \frac{1}{L}(1+L)^{1+1/L}.
$
\end{proof}

\subsubsection{Comparison of the bounds}
\label{subsec:bidding.comparisons}

In the Appendix we prove that the ratio between the two bounds is approximately
\[
\log \frac{\text{UB}}{\text{LB}} \leq \frac{\sqrt{8k\tau(1-\tau)} k(1-\tau)(1-{\cal H}(\frac{\tau}{1-\tau}))}
{2^{k(1-\tau)(1-{\cal H}(\frac{\tau}{1-\tau}))}} -
\frac{k(1-{\cal H}(\tau))}{2^{k(1-{\cal H}(\tau))}},
\]
where $\tau=H/k$. We infer that as $k$ increases, and for any fixed value of $\tau$, the upper and lower bounds become very close to each other.


\section{Online fractional knapsack}
\label{sec:fractional.knapsack}

In the online fractional knapsack problem, the request sequence consists of {\em items}, where item $i$ has a {\em value} $v_i \in \mathbb{R}^+$ and a {\em size} $s_i \in (0,1]$. The online algorithm has a knapsack of unit capacity, and when considering item $i$, it can accept irrevocably a fraction 
$f_i \in (0,1]$ of the item, subject to capacity constraints. More precisely, the algorithm aims to maximize $\sum\limits_{i} (f_i \cdot v_i)$ subject to $\sum\limits_{i} (f_i\cdot s_i) \leq 1$.

Let $d_i=v_i/s_i$ denote the {\em density} of item $i$. While the offline version of the problem admits a simple, optimal solution via a greedy algorithm (that sorts all items by non-decreasing order of density, and accepts items in this order until the knapsack is full), the online version is more challenging. Suppose that $d_i \in [L,U]$, for $L,U$ known to the algorithm. 
\cite{buchbinder2005online, buchbinder2006improved} gave matching $O(\log (U/L))$ and $\Omega(\log (U/L))$ upper and lower bounds on the competitive ratio of the problem, respectively, and~\cite{zhou2008budget} showed an optimal bound of $\ln (U/L)+1$ for deterministic algorithms. Online fractional knapsack has applications in sponsored search auctions, and online ad allocation, and has been studied in several other settings, e.g.,~\cite{albers2021improved,kesselheim2014primal}.
In this section, we study this problem in the imperfect advice setting.

\subsection{Upper bound}
\label{subsec:knapsack.upper}

As in all previous work, we assume that the density of all items is in $[L, U]$ for known values of $L$ and $U$. Let $d^*$ denote the smallest density of an item included at a positive fraction in the optimal solution. That is, the optimal algorithm \Opt accepts a fraction 1 of items with density larger than $d^*$, and fills the remaining space with a fraction of items of density $d^*$. Unfortunately, knowing $d^*$ (even its exact value) is not sufficient for an online algorithm to be anywhere as efficient as \Opt. For example, an algorithm that accepts a fraction 1 of items of density larger than $d^*$  has unbounded competitive ratio in sequences that consist only of items of density 
$d^*$.  Similarly, an algorithm that accepts a fraction 1 of items with density at least $d^*$ 
has unbounded competitive ratio in sequences in which items of density $d^*$ appear early in the sequence, and items of greater density later in the sequence.
However, if we denote by $c^* \in (0,1)$ the fraction of the knapsack in the optimal solution that is either empty, or occupied with items of density $d^*$, then knowing the exact value of both $d^*$ and $c^*$ suffices to achieve optimality. Our approach will then aim to use $k$ comparison queries so as to approximate the values of $c^*$ and $d^*$, then use these approximations to choose fractional items.

\subsubsection{Algorithm and analysis} 
We describe the online algorithm. We first define two types of partitions, related to the parameters $d^*$ and $c^*$. In what concerns $d^*$, 
partition the interval $[L,U]$ into $s$ 
sub-intervals $I_1,\ldots, I_{s}$ such that $I_i = [d_{i-1},d_{i})$, for $s$ that will be specified later. We also set $L=d_0, U=d_{s}$. The values $d_i$ are defined so that:
$\beta = \frac{d_1}{d_0} = \frac{d_2}{d_1} = \ldots = \frac{d_{s}}{d_{s-1}}$.
Thus, we have $\beta = (U/L)^{1/s}$ and $d_i = L\cdot \beta^{i}$, and note that 
$d^* \in I_x$ for some $x \in [1,s]$. 

In what concerns the parameter $c^*$, we partition the interval $[0,1]$ 
into $m$ sub-intervals $I'_1,\ldots, I'_{m}$ such that 
$I'_i = [c_{i-1},c_{i})$; we have $c_0 = 0$ and $c_{m} = 1$. The value of $m$ will be determined later; the values $c_i$ are defined so that
$c_1 = c_2 - \frac{c_1}{\beta} = c_3 - \frac{c_2}{\beta} = \ldots = c_m - \frac{c_{m-1}}{\beta}. $

It readily follows that for $i\geq 1$, we have $c_i = 
\frac{\beta^{m+i-1}-\beta^{m+i-2}}{\beta^{m}-1}$.
In particular, $c_1=\frac{\beta^{m}-\beta^{m-1}}{\beta^{m}-1}$, and $\frac{1}{1-c_1} = \frac{\beta^{m}-1}{\beta^{m-1}-1} $. Note also that $c^* \in I'_y$ for some $y \in [1,m]$.

Provided that $s\cdot m \leq \lfloor 2^{k-H}/\multiset{k-H}{H} \rfloor$, 
Theorem~\ref{th:range} shows that the algorithm can use $k$ comparison queries 
so as to identify both $x$ and $y$.
Given these values, the algorithm reserves, in its knapsack, 
a capacity $c = c_{y-1}$ for items with density in the range $I_x = [d_{x-1},d_x)$, to which we refer as \emph{critical items}. The algorithm uses the remaining capacity of $1-c$ for items of density larger than $d_x$, to which we refer as \emph{heavy items}. The algorithm accepts a fraction 1 of all critical items, as long as the capacity $c$ reserved for them allows. Similarly, the algorithm accepts a fraction 1 of heavy items and places them in their dedicated space of the knapsack. Given that $c^* \in I_y$, we have $1-c > 1-c^*$; that is, the reserved capacity for heavy items is at least equal to the total size of these items. In other words, the algorithm can afford to accept all heavy items. The algorithm rejects all items of density smaller than $d_{x-1}$. 

\begin{theorem}[Appendix]
For any $H\leq k/2$, the above algorithm has competitive ratio at most
\vspace{-0.3cm}

  \begin{align*}
    \min\limits_{s,m \in \mathbb{N}} \ \ \ \ f_m(\beta) \ \ \  \ \ \ &  \ {\emph{where \ \ \ }}  \beta = (U/L)^{1/s}, 
    \ \ \emph{and } \ 
    f_m(\beta)= \frac{\beta^{m}-1}{\beta^{m-1}-1} \\ 
  &  {\emph{subject to \ \  }}  s\cdot m \leq \lfloor 2^{k-H}/\multiset{k-H}{H} \rfloor. \\
  \end{align*} 
  \label{thm:knapsack.upper}
\end{theorem}

\vspace{-1.4cm}
\subsection{Lower bound}
We will show a lower bound $C(k,H)$ on the competitive ratio of any algorithm with imperfect advice. For the sake of contradiction, suppose there is an algorithm $A$ of competitive ratio better than
$C(k,H)$. Our proof is based on a reduction from the \Range game. Specifically, we prove that, based on $A$, we obtain a questioner's strategy for \Range which can find a value $z\in\{1,\ldots, p\}$, with $p=\lceil 2^k/\multiset{k}{H}\rceil+1$, which contradicts Theorem~\ref{th:range}.

We give the intuition behind the proof. 
Let $s$ and $m$ be any two positive integers such that $s\cdot m \leq p$ and $s\cdot (m+1) > p$.
Define $\beta = (U/L)^{1/s}$, and $d_i = U\cdot \beta^i$, for $i\in [1,s]$.
Given a pair $(x,y)$ of integers, where $x\in \{1,\ldots, s \}$ and $y\in \{1,\ldots m+1\}$, define the sequence $$\sigma_{x,y} = ( (d_1,1),(d_2,1),\ldots, (d_{x-1},1),(d_{x},c_y),$$
where $(d_i,j)$ indicates a subsequence of $j/\epsilon$ items, each of which has size $\epsilon$
and density $d_i$, and where $\epsilon$ is infinitesimally small. $c_y \in [0,1]$ is defined appropriately
in the proof. For this sequence, $\text{OPT}(\sigma_{x,y})=(1-c_y)d_{x-1}+c_yd_x$. There 
are $s\cdot (m+1) > p$ such sequences, and $\sigma_{x,y}$ is a prefix sequence of $\sigma_{x,y+1}$, and $\sigma_{x,m}$ is a prefix sequence of $\sigma_{x+1,1}$. In the proof, we consider request sequences of this form, and we show that if $A$ is $C(k,H)$-competitive, its decisions can help find any given $z\in\{1,\ldots,p\}$, which contradicts Theorem~\ref{th:range}. We refer to Appendix for the technical details.

\begin{theorem}[Appendix]
For the fractional knapsack problem, where items densities are in $[L, U]$, no deterministic algorithm with $k$ subset queries, out of which $H\leq k/2$ may have erroneous responses, can achieve a competitive ratio better than
\vspace{-0.2cm}
\begin{align*}
C(k,H) = \min\limits_{s,m \in \mathbb{N}} \ \ \ \ g_m(\beta) \ \ &  \ {\emph{where \ \ \ }}  \beta = (U/L)^{1/s},
\  \  g_m(\beta) = (\frac{\beta^2-\beta+1}{2\beta+1})^{1/(m+1)} \\
\hspace*{2cm} {\emph{subject to \ \  }} & s\cdot m \leq \lceil 2^{k}/\multiset{k}{H} \rceil+1.\\
\end{align*}
\label{thm:knapsack.lower}
\end{theorem}

\vspace{-1.5cm}
\paragraph{Comparison of the bounds}
Let $\tau=H/k$.  Since $\frac{\beta^{m}-1}{\beta^{m-1}-1} \leq \beta$, using~\eqref{eq:amazing}, the upper bound of Theorem~\ref{thm:knapsack.upper} is at most
$(U/L)^q$, where $q\leq 1/2^{k(1-\tau)(1-{\cal H}(\frac{\tau}{1-\tau}))}$. Furthermore, since 
$\frac{\beta^2-\beta+1}{2\beta+1} \geq \frac{\beta}{3}$ (for all $\beta\geq 3$), the lower
bound of Theorem~\ref{thm:knapsack.lower} is at least $(U/L)^{q'} (1/3)^{q'}$, where 
$q' \geq 1/(2\sqrt{8k\tau(1-\tau)}2^{k(1-{\cal H}(\tau))}+1)$, for all $U/L \geq 3$. For simplicity, we omitted the floors and ceilings.

\section{Waiving the assumption of the tolerance parameter}
\label{sec:extensions}

In the imperfect advice setting we studied so far, the algorithm defines an application-specific tolerance parameter that measures its desired tolerance to errors (or equivalently, an anticipated upper bound on the error). This parameter is in a sense required, since the analysis of 
R\'enyi-Ulam games in~\cite{RivestMKWS80} involves the extreme value of error (i.e., $H$) instead of the instance-specific error value (i.e., $\eta$). Nevertheless, in this section, we discuss how to mitigate the need for pre-determining a tolerance parameter. We propose two different approaches, based on {\em resource-augmentation}, and {\em robustification}, which we discuss in what follows. We use the time-series search and online bidding problems as illustration, even though our approach may carry through in other online problems, at the expense of more complex calculations. 

\subsection{Resource augmentation}
\label{subsec:resource}

In this setting, we compare an {\em oblivious} online algorithm $A$ with $l$ advice bits and no information on the error bound, to an online algorithm $B$ that has $k$ {\em ideal} (i.e. error-free) advice bits. Specifically, we are interested in finding the smallest $l\geq k$ (as function of $k$) for which algorithm $A$ is at least as good as algorithm $B$, regardless of the error in the advice of $A$.

The following theorem shows that $O(1)$-factor resource augmentation suffices to obtain an oblivious algorithm that is at least as efficient as any algorithm that operates in the ideal setting of error-free advice, and even if a fraction $1/3-c$ of the advice bits may be erroneous, for any constant $c$. 

\begin{theorem}[Appendix]
Consider the time-series and the online bidding problems. For all sufficiently large $k$, and any $c\in (0,1/3)$, there is an oblivious online algorithm $A$ with advice of size $l$, whose competitive ratio is at 
least as good as that of {\em any} online algorithm $B$ with $k$ bits of perfect (i.e. error-free) advice, where 
$
l=\frac{1}{(\frac{2}{3}+c)(1-{\cal H}(\frac{\frac{1}{3}-c}{\frac{2}{3}+c}))} k+1,
$
for any error $\eta\leq (1/3-c)l$ in the advice of $A$.
\label{thm:resource.ts+bidding}
\end{theorem}

\subsection{Robustification}
\label{subsec:robustification}

In this setting, we augment the imperfect advice framework by requiring not only that the algorithm minimizes the competitive ratio assuming that the advice error is at most the tolerance $H$, but also that its competitive ratio does not exceed a {\em robustness} requirement $r$, for some specified $r$, if the error exceeds $H$ (and in particular, if the advice is adversarially generated). We call such online algorithms {\em $r$-robust}. Thus, this model can be seen as an extension of both the imperfect advice and the untrusted advice model of~\cite{DBLP:conf/innovations/0001DJKR20}.

For the time-series problem, we obtain the following result, which generalizes Theorem~\ref{thm:ts}. In particular, note  that Theorem~\ref{thm:ts} is a special case of Theorem~\ref{thm:ts.robust}
for $\rho=1$.

\begin{theorem}[Appendix]
Consider the online time series search problem, with imperfect advice of size $k$, tolerance $H\leq k/2$, and robustness $r=(M/m)^\rho$, where $\rho \in (1/2,1]$. There is an $r$-robust algorithm that uses $k$ comparison queries, and has competitive ratio at most 
$(M/m)^\frac{2\rho-1}{U+1}$, where $U={\lfloor 2^{k-H}/\multiset{k-H}{H}}\rfloor$, for any $H\leq k/2$. Moreover, every (deterministic) algorithm based on $k$ subset queries has competitive ratio
better than $(M/m)^\frac{2\rho-1}{L+1}$, where $L={\lceil 2^k/\multiset{k-H}{H}}\rceil$.
\label{thm:ts.robust}
\end{theorem}

The analysis of $r$-robust algorithms for online bidding is more challenging, in particular in what concerns the impossibility results. We give an overview of the approach. For the upper bound, we can follow an analysis along the lines of Theorem~\ref{thm:noisy.upper}, however, each bidding sequence in the collection ${\cal X}_{b,2^k}$ must be individually $r$-robust. This is easy to enforce, and it requires that $b$ much be such that $b^2/(b-1)\leq r$. The lower bound is more subtle: the proof follows the lines of Theorem~\ref{thm:noisy.lower}, but uses the fact that if all the $l$ sequences in $X_0, \ldots ,X_{l-1}$ must be $r$-robust, then $\alpha_{\bar X}^2/(\alpha_{\bar X}-1) \leq r$. 
We obtain the following:

\begin{theorem}[Appendix]
For every $r\geq 4$ there is an $r$-robust bidding strategy with $k$-bit imperfect advice that has competitive ratio at most 
\[
\min_{b>1} \frac{b^{2^k+U+1}}{b^{2^k}-1}, \quad \textrm{ subject to } b^{2^{k+1}}/(b^{2^k}-1) \leq r, \quad \textrm{ and where } U=\lceil 2^H \mch{k-H}{H} \rceil.
\]
Furthermore, every $r$-robust bidding strategy with $k$-bit imperfect advice has competitive ratio at least
\[
\min_{\alpha>1} \frac{\alpha^{2^k+L+1}}{\alpha^{2^k}-1} \quad \textrm{ subject to } \alpha^{2k}/(\alpha^k-1) \leq r, \quad \textrm{ and where } L=\lfloor \mch{k}{H} \rfloor. 
\]
\label{thm:bidding.robust}
\end{theorem}

\bibliographystyle{plain}
\bibliography{fault-contract,targets-arxiv,targets,refs}

\appendix

\section{Appendix}

\subsection{Omitted material of Section~\ref{sec:core}}

\begin{proof}[Proof of Theorem~\ref{thm:cyclic.upper}]
We will reduce \mincycl to the following game that was studied in~\cite{RivestMKWS80}:
\bigskip

\noindent \textbf{\Identify game:}
In this game, $x$ is an integer in  $\{0,1,\ldots, m-1\}$ for some known $m$, and the objective is to identify $x$ with as few queries as possible, if up to $H$ queries may be answered incorrectly. 

We will use the following result in the analysis:
\begin{lemma}~\cite{RivestMKWS80} \label{lem:identifyGame}
The number $Q(m,H)$ of queries required to identify $x$ in an instance of \Identify is such that 
\[
\min\{k'|2^{k'}\geq m\cdot \multiset{k'}{H} \} \leq Q(m,H) \leq \min \{k| 2^{k-H} \geq m \multiset{k-H}{H}\}.
\]
\end{lemma}

Given an instance of \mincycl, we create an instance of \Identify, with 
$m = 2^{k-H}/\mch{k-H}{H}$ (note that $H$ is the same for both instances). Partition the interval $[0, \ldots, n-1]$ into $m$ disjoint subintervals, each of length at most $\lceil n/m \rceil$. Let $x$ being the index in $[0, \ldots ,n-1]$ for which $A[x]=0$, and let $I_x$ denote the interval that contains $x$. 
A query $q$ of the \wa strategy of~\cite{RivestMKWS80} that asks ``is $I_x \leq b$ for some $b \in \{0,m-1\}$?" 
translates to query $\mu(x)$ in the \mincycl instance that asks ``is $x \leq f(b)$?", where $f(b)$ is the largest value in the interval $I_b$. Note that the answer to $q$ is `yes' if and only if the answer to $\mu(q)$ is `yes'. The response to $\mu(q)$ is then given to \wa, which updates its state and proceeds with the next query. For $m$ defined as above, we have that 
$2^{k-H} \geq m \cdot \mch{k-H}{H}$ and using \wa, we can find $I_x$ using $k$ queries.  
Subsequently, we return the largest integer $f(x)$ in $I_x$. 
Given that $x$ is in $I_x$ and the length of the intervals is at least  $\lceil n/m \rceil$, we conclude that the returned index $j$ is such that $A[j]=\leq \lceil n/m \rceil = \lceil n \mch{k-H}{H}/2^{k-H} \rceil $.
\end{proof}

\begin{proof}[Proof of Theorem~\ref{thm:search.lower}]
Consider the following game that is defined as the \cSearch game, with the only difference that the goal is to guess a value as close to some 
$r \in [0,n)$, for some fixed $n$ (for the purpose of the proof, we can think of $n$ as sufficiently large). After receiving the responses to the $k$ queries, the questioner returns a number $r' \in [0,n)$, and the objective is to minimize $|r'-r|$. We will show a reduction from this game that will help us establish the lower bound. 
Suppose, by way of contradiction, that there exists a strategy, say ALG for \search that returns an element $e$ with $\pi(e) < \lfloor n\mch{k}{H}/2^{k} \rfloor - 1$. We devise a strategy for the questioner in the continuous game based on ALG. Given values of $r$, $k$ and $H$ that define an instance $G$ of the continuous game over the continuous interval $I = [0,n)$, create an instance $S$ of \search on a space $A$ of $n$ elements, with the same values of $k$ and $H$. Consider a bijective mapping $\beta$ that maps an element of rank $i$ in $A$ ($i\in \{0, \ldots, n-1\}$) to an interval $\beta(i) = [i,i+1)$ in $I$. 
Similarly, define a bijective mapping $\mu$ between queries asked for $G$ and those asked for $S$. Any range $[i,j]$ of indices that is a part of a subset query $q$ asked for $S$ is mapped to an interval $[i,j+1)$ in the query $\mu(q)$ asked for $G$. 
Let $r$ denote the searched value in $G$ and let $x$ denote an index of $A$ such that $r$ belongs to $\beta(x)$.
To search for $r$, we consider queries that ALG asks for $S$ and for any such query $q$, we ask $\mu(q)$ for $G$. The response to $\mu(q)$ is then given to ALG so that it can update its state and ask its next query.

Recall that we supposed that  ALG outputs an element $e$ such that $\pi(e) < \lfloor n\mch{k}{H}/2^{k} \rfloor -1$, and that $\beta(\pi(e)) = [\pi(e),\pi(e)+1)$. As an output for $G$, we return $r' = \pi(e)+1$ as the answer for $G$. Note that there are exactly $e+1$ intervals from the range of $\beta$ that lie between $r$ and $r'$ in $I$. That is $|r-r'| < \lfloor n\mch{k}{H}/2^{k} -1\rfloor+1 \leq n(\mch{k}{H}/2^{k}$. This, however, contradicts Lemma~\ref{lem:csearch}. 
\end{proof}

\subsection{Omitted material of Section~\ref{sec.bidding}}

\begin{proof}[Proof of Theorem~\ref{thm:noisy.upper}]
We apply the algorithm of Theorem~\ref{thm:cyclic.upper} on the set of indices of all 
sequences in ${\cal X}_{b,2^k}$ with $n=2^k$, where $b>1$ will be chosen later. The output is the index of a strategy in ${\cal X}_{b,2^k}$ which is ranked at most $U$ among the sequences in ${\cal X}_{b,2^k}$ 
From the definition of ${\cal X}_{b,2^k}$, and in particular its cyclic property, this means that the selected strategy discovers the target with cost at most $b^U$ times larger than the best strategy in ${\cal X}_{b,2^k}$. We infer that the competitive ratio of the chosen strategy is at most $\frac{b^{2^k+1+U}}{b^{2^k}-1}$.  This expression is minimized for $b=(\frac{2^k+U+1}{U+1})^{1/2^k}$, from which we obtain that the competitive ratio of our strategy is at most
\[
\frac{1+U}{2^k}\left(1+\frac{2^k}{1+U}\right) ^{1+\frac{1+U}{2^k}}.
\]
\end{proof}

\begin{proof}[Proof of Theorem~\ref{thm:multi-search.alpha.faulty}]
Consider a $p$-parallel strategy $X$, defined by $p$ bidding strategies $X_0, \ldots ,X_{p-1}$, each run on a dedicated processor. Let $x_{j,i}$ denote bid $i$ in $X_j$; we say that $x_{j,i}$ {\em precedes} bid $x_{j,i'}$ in $j$ if $i<i'$. We define the {\em prefix cost} of a bid in $X_j$ as the sum of the values of all bids that precede that bid in $X_j$. For $j \in [0,\ldots ,p-1]$, we denote by $u_X(c,j)$ as the value of the largest bid in $X_j$, such that the sum of the prefix cost of that bid and the value of that bid do not exceed $c$. We also define by $u_{X,\phi}(c)$ at the $(\phi+1)$-largest quantity in the set $\{u_X(c,j)\}_{j=0}^{p-1}$. 

 Consider an arbitrary indexing of all bids in $X$, i.e., the $i$-th bid is such that it is the $m$-th bid in $X_j$, for some $m,j$. We will represent this bid as a pair of the form $(C_i,D_i)$, where $C_i$ is the cost of all bids that precede bid $i$ in the sequence to which it belongs, and $D_i$ is the bid itself (i.e., its value). Note that this representation ignores the specific sequence to which the bid is assigned, since this is not important for the purposes of the proof, as we will see. Given a bid represented as $(C_i,D_i)$ we define $d_i$ to be equal to  $u_{X,\phi}(C_i+D_i)$: we call this value the $(\phi+1)$-largest bid relative to $D_i$, in $X$.

Recall that $\bar{X}$ denotes the sequence of all bid values in $X$, in non-decreasing order. 
Hence, each bid in $X$ is mapped via its length to an element of this sequence (breaking ties arbitrarily).

Fix a bid $i_0$ of the form $b_{i_0}=(C_{i_0},D_{i_0})$, and suppose, without loss of generality, that $b_{i_0}$ belongs to sequence $X_0$. Let $c=C_{i_0}+D_{i_0}$. For all $m \in [1, p-1]$, let 
$b_{i_m}=(C_{i_m},D_{i_m})$ denote the largest bid in $X_m$ for which the sum of its prefix cost and the value of its bid are at most $c$. For every $m \in [0,p-1]$, define $I_m$ as the set of indices in ${\mathbb N}$ such that $i \in I_m$ if and only if a bid of value $x_i$ is such that its prefix cost plus $x_i$ does not exceed $c$. From the definition of the competitive ratio we have that 
\[
\comp(X) \geq \frac{\sum_{i \in I_m} x_i}{d_{i_m}}, \qquad \ \textrm{ for all $m \in [0, p-1]$.}
\]
Therefore,
\[
\comp(X) \geq \max_{0 \leq m \leq p-1} \frac{\sum_{i \in I_m} x_i}{d_{i_m}},
\]
and using the property $\max\{ a/b, c/d \} \geq \frac{a+b}{c+d}$, for all $a,b,c,d>0$, we obtain that
\begin{equation}
\comp(X) \geq  \frac{\sum_{m=0}^{p-1}\sum_{i \in I_m} x_i}{\sum_{m=0}^{p-1} d_{i_m}}.
\label{eq:summand}
\end{equation}
Next, we will bound the numerator of the fraction in~\eqref{eq:summand} from below, and its denominator from above. We begin with a useful observation: we can assume, without loss of generality, that for cost $c$ (defined earlier), no bid of value $d_{i_0}$ or smaller has prefix cost larger than 
$c$ minus the value of the bid in question. This follows from the definition of $d_{i_0}$: if such a bid existed, then one could simply ``remove'' this bid from $X$, and obtain a $p$-parallel sequence of no worse competitive ratio (in other words, such a bid is useless, and one can derive a sequence of no larger competitive ratio than $X$ that does not contain it). 

Using the above observation, it follows that the numerator in~\eqref{eq:summand} includes, as summands, all bids of value at most $d_{i_0}$, as well as at least $\phi+1$ bids that are at least as large as $d_{i_0}$ ($\phi$ of those bids are from the definition of $d_{i_0}$, and the additional one is bid $b_{i_0}$). Let $q$ denote an index such that $d_{i_0}=\bar{x}_{q}$, then we have that 
\[
\sum_{m=0}^{p-1}\sum_{i \in I_m} x_i \geq \sum_{i=0}^{q+\phi+1} \bar{x}_i.
\]
We now show how to upper-bound the denominator, using the monotonicity implied in the definition of 
the $(\phi+1)$-largest value relative to a given bid value, and the definition of the bids 
$b_{i_0}, \ldots , b_{i_{p-1}}$. Specifically, for  every bid $b_{i_m}$, with $m \in [1, p-1]$, we have that $d_{i_m} \leq d_{i_0}$. It thus follows
that
\[
\sum_{m=0}^{p-1} d_{i_m} \leq \sum_{i=q}^{q-(p-1)} \bar{x}_{i}.
\]
Combining the two bounds, it follows that 
\[
\comp(X) \geq \sup_{0\leq q<\infty} \frac{\sum_{i=0}^{q+\phi+1} \bar{x}_i}{\sum_{i=q}^{q-(p-1)} \bar{x}_{i}}.
\]

In the last step of the proof, we will use a result from search theory, namely Gal's {\em functional} theorem, stated below:

\begin{theorem}[Gal~\cite{gal:general}]
Let $q$ be a positive integer, and $X=(x_i)_{i=0}^\infty$ a sequence of positive numbers with 
$\sup_{n \geq 0} x_{n+1}/x_n <\infty$ and $\alpha_X>0$. Suppose that $F_i$ is a sequence of functionals that satisfy the following properties:
\begin{itemize}
\item[(1)]$ F_i(X)$ depends only on $x_0,x_1, \ldots x_{i+q}$,
\item[(2)] $F_i(X)$ is continuous in every variable, for all positive sequences $X$,
\item[(3)] $F_i(a X)=F_i(X)$, for all $a>0$,
\item[(4)] $F_i(X+Y) \leq \max(F_i(X), F_i(Y))$, for all positive sequences $X,Y$, and
\item[(5)] $F_{i+j}(X) \geq F_i(X^{+j})$, for all $j \geq 1$, where $X^{+j}=(x_j, x_{j+1}, \ldots)$.  
\end{itemize}
Then 
\[
\sup_{0 \leq k <\infty} F_k(X) \geq \sup_{0 \leq k<\infty} F_k (G_{\alpha_X}), 
\]
where $G_a$ is defined as the geometric sequence $(a^i)_{i=0}^\infty$.
\label{thm:gal}
\end{theorem}

Define now the functional $F_{q}(\bar{X})=\frac{\sum_{i=0}^{q+\phi+1} \bar{x}_i}{\sum_{i=q}^{q-(p-1)} \bar{x}_{i}}$, for every $q$. The functional satisfies the conditions (1)-(5) of 
Theorem~\ref{thm:gal} (see Example 7.3 in~\cite{searchgames}). By applying Gal's Theorem, it follows that 
\[
 \comp(X) \geq 
 \sup_{0 \leq q<\infty} 
 \frac {    \sum_{i=0}^{q+\phi+1}  \alpha_{\bar{X}}^i   } 
{\sum_{i=q}^{q-(p-1)} \alpha_{\bar{X}}^i }.
\]
If $\alpha_{\bar{X}}\leq 1$, then it is easy to show that the above expression shows that $\comp(X)=\infty$; see, e.g.~\cite{aaai06:contracts}. Otherwise, i.e., if $\alpha_{\bar{X}}>1$, after some simple calculations 
we arrive at the desired result.
\end{proof}

\begin{proof}[Comparison between the upper and the lower bounds]

We compare the upper and lower bounds of Sections~\ref{subsec:bidding.upper} and~\ref{subsec:bidding.lower}. Define $f$ as the function
$f(x)=\frac{1}{x}(1+x)^{1+1/x}$, and note that $f$ is decreasing in $x$, with $f(1)=4$, and $\lim_{x \to \infty} f(x)=1$. Then, the upper bound of Theorem~\ref{thm:noisy.upper} is equal to $f(2^k/(U+1))$, whereas the lower bound of Theorem~\ref{thm:noisy.lower} is equal to $f(L)$, where 
$U,L$ are defined in the statements of the corresponding theorems. 

For every $y>x$ we have
\[
\frac{f(y)}{f(x)} =\frac{\frac{1}{y}(1+y)^{1+1/y}}{\frac{1}{x}(1+x)^{1+1/x}} \leq \frac{(1+y)^{1/y}}{(1+x)^{1/x}}.
\]
Moreover, using some more elementary calculus, 
\[
\log \frac{f(y)}{f(x)} \leq \log \frac{(1+y)^{1/y}}{(1+x)^{1/x}} \leq \log  \frac{y^{1/y}}{x^{1/x}}=\frac{1}{y}\log y-\frac{1}{x}\log x .
\]

We will use the above inequality to compare $f(2^k/(U+1))$ to $f(L)$. To simplify the calculations, we will assume that the upper bound is $f(2^k/U)$, since the additive ``one'' in the numerator has virtually no effect on the competitive ratio as $k$ becomes large. For the same reasons, we ignore the ceiling in the expression of $U$. Using the approximation of the partial sum of binomial coefficients of~\eqref{eq:amazing}, and defining $\tau=k/H$, we obtain that the ratio UB/LB of the upper and lower bounds satisfies
\[
\log \frac{\text{UB}}{\text{LB}} \leq \frac{\sqrt{8k\tau(1-\tau)} k(1-\tau)(1-{\cal H}(\frac{\tau}{1-\tau}))}
{2^{k(1-\tau)(1-{\cal H}(\frac{\tau}{1-\tau}))}} -
\frac{k(1-{\cal H}(\tau))}{2^{k(1-{\cal H}(\tau))}},
\]
\end{proof}

\subsection{Omitted material of Section~\ref{sec:fractional.knapsack}}

\begin{proof}[Proof of Theorem~\ref{thm:knapsack.upper}]
	Note that a fraction 1 of heavy items is accepted by both the online algorithm and \Opt. Therefore, the contribution of heavy items to the profits of the algorithm and \Opt are the same, say $\Delta$; we have $\delta \geq (1-c^*) d_x \geq (1-c_y) d_{x} $; this is because $c^*\in [c_{y-1},c_{y})$ and all heavy items have density larger than $d_x$. 
	The algorithm fills the reserved space of size $c$ with critical items, which are of density at least $d_{x-1}$, while \Opt fills a space of $c^* < c_y$ with critical items, which are of density at most $d_x$. 
	Thus, we have
	\begin{align*}
		\comp(A) \leq \frac{\Delta + c_y \cdot d_{x}}{\Delta + c_{y-1} \cdot d_{x-1}} & 
		\leq \frac{ (1-c_y) d_{x} + c_y \cdot d_{x}}{ (1-c_y) d_{x} + c_{y-1} \cdot d_{x-1}} & \\
		& = \frac{ 1}{1- (c_y - c_{y-1}/\beta)} = \frac{1}{1-c_1} = \frac{\beta^{m}-1}{\beta^{m-1}-1}. &\\   
	\end{align*}
\end{proof}

\begin{proof}[Proof of Theorem~\ref{thm:knapsack.lower}] \ \
	By way of contradiction, suppose there is an algorithm A with a better competitive ratio
  than $C(k,H)$. We will show that A could then be used in the \Range game with $k$ queries so as to identify an unknown value in $\{1,\ldots, p\}$, 
	which contradicts the upper bound of Theorem~\ref{th:range}.
	
	Fix the values of $(s,m')$ that minimize $g_m(\beta)$ subject to $s.m'\leq p$, and let $m = m'+1$. Define $p = s\cdot m$; we have $p > \lceil 2^k/\multiset{k}{H} \rceil$, otherwise, the pair $(s,m)$ results in a smaller value for $g_m(\beta)$ (note that $g_m(\beta)$ is a decreasing function of $m$).
	Suppose we want to identify an unknown value $z\in \{1,\ldots,s\cdot m\}$; this is equivalent to finding a pair $(x,y)$ with $x\in \{1,\ldots s\}$ and $y\in \{1,\ldots, m\}$. 
	For $i\in \{1,\ldots,s\}$, let $d_i = L\cdot \beta^{i}$.  Moreover, for $j\in \{0,\ldots,m \}$, define $c_j$ such that the following hold: 
	\begin{align*}
		c_0  = \min\{1/\beta,1-1/\beta\}, \ \ \  
		c_m  = \max\{1/\beta, 1-1/\beta\}, & \\   r = \frac{c_1+\beta(1-c_1)}{c_0+\beta(1-c_0)} =  \frac{c_2+\beta(1-c_2)}{c_1+\beta(1-c_1)} = 
		\ldots = 
		\frac{c_{i+1}+\beta(1-c_{i+1})}{c_{i}+\beta(1-c_{i})} = \ldots = & \frac{c_{m-1}+\beta(1-c_{m-1})}{c_{m}+\beta(1-c_{m})}.
	\end{align*}
	
	It can be verified that $r = (\frac{\beta^2-\beta+1}{2\beta+1})^{1/m}$. 
	For any $x\in\{1,\ldots, s\}$ and $y\in\{1,\ldots,m\}$, create an input sequence as $$\sigma_{x,y}  = (1,d_0), (1,d_1), (1,d_2), \ldots, (1,d_{x-1}), (c_{y},d_{x}).$$ Here $(a,b)$ indicates a sequence of items, all of infinitesimal small size $\epsilon$ and density $b$, and total size $a$. The optimal solution fills a capacity $c_{y}$ with the item of density $d_{x}$, and the remaining capacity of $1-c_{y}$ with the item of density $d_{x+1}$. We have $$\text{OPT}(\sigma_{x,y}) = (1-c_{y})\cdot d_{x-1} + c_{y} \cdot d_{x}.$$
	
	In what follows, we describe how an algorithm A with a competitive ratio better than $C(k,H)$ can be used to correctly find unknown values $x\in\{1,\ldots s\}$ and $y\in\{1,\ldots,m\}$.
	Suppose the next item has density $d_\alpha$. Let $w_\alpha$ denote the total size of items of density $\leq d_{\alpha-1}$, and suppose $w_\alpha\in [1-c_q,1-c_{q-1})$ for some $q\in \{[1,\ldots, m\}$, that is, $1-w_\alpha \in (c_{q-1},c_{q}]$.  
	Define $\Delta_\alpha = c_{q+1}-c_q$. 
	When the empty space in the knapsack of A becomes less than $\Delta_\alpha$, the algorithm ``guesses" $x = \alpha$ and $y =q$. 
	In what follows, we show that A makes these guesses at some point and the guesses made by A are correct. 
	 For the sake of contradiction, suppose A does not make a guess, or at least one of its guesses is incorrect. We show that the competitive ratio of A will be larger than $C(k,H)$. Suppose $w_x\in [1-c_q,1-c_{q-1})$ for some $q\in \{1,\ldots,m\}$, that is, $1-w \in (c_{q-1},c_{q}]$. 
	There are four possibilities to consider:
	
	\begin{itemize}
		\item[] \textbf{A does not make a guess:} Since A does not make a guess, the empty space in the knapsack is at least $\Delta_x = c_{q+1}-c_q$. Also, the total size of items of density $\leq d_{x-1}$ is at least $1-c_q$. Therefore, 
			the contribution of items of density $d_x$ to the value of the knapsack is at most $c_q\cdot d_x$, and the contribution of other items is at most $(1-c_q-\Delta_x)d_{x-1} = (1-c_{q+1})d_{x-1}$. 
Therefore, the total value of value of the knapsack of A is at most $(1-c_{q+1})d_{x-1} + c_q d_x \leq (1-c_{y+1})d_{x-1} + c_{y}d_x$. We can write:
		
		$$\comp(A) \geq \frac{(1-c_y)d_{x-1}+c_y d_x}{(1-c_{y+1})d_{x-1} + c_{y}d_x} >  r.$$
		
		\item[]	\textbf{Wrong guess for $y$:} 
		Suppose the algorithm stops but makes the wrong guess for $y$, that is, $q \neq y$. First, suppose $q \leq y-1$, that is, A reserves too little space for items of density $d_x$. The total size of items of density $\leq d_{x-1}$ is at least $1-c_q$. Therefore, the total size of items of density $d_x$ is at most $c_q$. The final value of the knapsack is thus at most $(1-c_q)d_{x-1} + c_q d_x \leq (1-c_{y-1})d_{x-1}+c_{y-1} d_x$. We can write: 
		$$ \comp(A) \geq \frac{(1-c_y)d_{x-1}+c_y d_x}{(1-c_{y-1})d_{x-1}+c_{y-1} d_x} = r.$$

		Next, suppose $p \geq y+1$, that is, too much space is reserved for items of density $d_x$ and some of this space stays empty. The total size of items of density at most $d_{x-1}$ in the knapsack will be at most
		$1-c_{q+1}$, and the total size of items of density $d_x$ is at most $1-c_q$. Therefore, we have $A(\sigma_{x,y}) \leq (1-c_{q+1})d_{x-1} + c_q d_x \leq (1-c_{y+1})d_{x-1}+c_{y}d_x$. We can write 
		
		$$ \comp{(A)} \geq \frac{(1-c_y)d_{x-1}+c_y d_x}{(1-c_{y+1})d_{x-1} + c_{y}d_x } > r.$$
		\item[] \textbf{Wrong guess for $x$:}
		Suppose $p=y$ but $\alpha < x$ (note that $\alpha$ cannot be larger than $x$). 
		The value of the knapsack is maximized when the algorithm fills a capacity $c_y$ with items of density $\alpha$ and the rest with items of density $d_{\alpha-1}$, that is, $A(\sigma_{x,y})\leq (1-c_y)d_{\alpha-1} + \beta c_y d_{\alpha} \leq (1-c_y)d_{x-2} + \beta c_y d_{x-1}$. We can write 
		
		$$ \comp{(A)} \geq \frac{(1-c_y)d_{x-1}+c_y d_x}{(1-c_y)d_{x-2} + \beta c_y d_{x-1}} = \beta > r. $$
	\end{itemize}
	
	To summarize, as long as $\comp{(A)} < g_{m}(\beta)$, one can use A to guess both values of $x$ and $y$ correctly, that is, it can identify $z\in \{1,\ldots,p\}$ with $k$ queries. 
	This, however, contradicts the lower bound of Theorem~\ref{th:range}. Therefore, the initial assumption about the competitive ratio of A does not hold, and we conclude that $\comp{(A)} \geq g_m(\beta)$.
\end{proof}

\subsection{Omitted material of Section~\ref{sec:extensions}}

\begin{proof}[Proof of Theorem~\ref{thm:resource.ts+bidding}]
Consider the algorithms (upper bounds) of Theorems~\ref{thm:ts} and~\ref{thm:noisy.upper}, with $l$ imperfect advice bits. The advice error $\eta$ is, from the assumption at most $(1/3-c)l$, thus at most a fraction equal to $1/3-c$ of the advice bits may be erroneous. Then, from 
the discussion in Sections~\ref{subsec:ts.comparisons} and~\ref{subsec:bidding.comparisons}, it follows that these two algorithms have better competitive ratio (for the corresponding problem) than any algorithm with $k$ bits of advice (irrespectively of the latter's advice error), as long as 
\[
l(1-(1/3-c))(1-{\cal H}(\frac{1/3-c}{1-(1/3-c)})>k,
\] 
for all sufficiently large $k$, which yields the result.
\end{proof}

\begin{proof}[Proof of Theorem~\ref{thm:ts.robust}]
Define $p_1$ and $p_2$ to be such that $M/p_1=(M/m)^\rho$, and  $p_2/{m}=(M/m)^\rho$, respectively, and note that $m\leq p_1 \leq p_2 \leq M$. Then an algorithm for time-series search is $r$-robust if and only if it sets its reservation price equal to some $p \in [p_1,p_2]$. The proof proceeds along the lines of the proof of Theorem~\ref{thm:ts}; instead of using the queries to find a suitable reservation price in
$[m,M]$, we search instead for a reservation price in $[p_1,p_2]$. In particular, note that by definition of $p_1,p_2$, we have that  $p_2/p_1= (M/m)^{2\rho-1}$.
\end{proof}

\begin{proof}[Proof of Theorem~\ref{thm:bidding.robust}]
For the upper bound, the proof is similar to that of Theorem~\ref{thm:noisy.upper}. The only difference is that $b$ must be optimized under the condition that each strategy in ${\cal X}_{b,2^k}$ must be individually $r$-robust. We know that a geometric strategy for online bidding with base $b$
has competitive ratio at most $b^2/(b-1)$. Since each strategy is geometric with base $b^{2^k}$, it follows that as a long as $b^{2^{k+1}}/(b^{2^k}-1) \leq r$, every strategy in ${\cal X}_{b,2^k}$ is 
$r$-robust, hence the strategy chosen by the imperfect advice as well. 

For the lower bound, we appeal to the following property shown in~\cite{DBLP:journals/corr/abs-2111-05281}: If all $l$ bidding sequences in the collection $X_0, \ldots ,X_{l-1}$ are $r$-robust, then it must be that $\alpha_{\bar X}^2/(\alpha_{\bar X}-1) \leq r$, where recall that $\bar{X}$ is the sequence of the union of all bids in the $l$ strategies, sorted in non-decreasing order. Then the proof follows along the lines of the proof of Theorem~\ref{thm:noisy.lower}, with the observation that the competitive ratio is minimized if $l=2^k$.
\end{proof}

\end{document}